\documentclass[a4paper, 10pt, reqno]{amsart}

\usepackage[utf8]{inputenc}
\usepackage[T1]{fontenc}
\usepackage[largesc]{newtxtext}	
\usepackage{newtxmath}
\usepackage{microtype}
\usepackage{amsmath}
\usepackage{amsthm}

\usepackage{amssymb}
\usepackage{amsfonts}
\usepackage{thmtools}
\usepackage{relsize}
\usepackage[mathscr]{eucal}
\usepackage[linktocpage]{hyperref}
\usepackage[sort,capitalize]{cleveref}
\usepackage{enumitem}			
\usepackage[dvipsnames]{xcolor}
\usepackage[msc-links, abbrev, non-sorted-cites]{amsrefs}
\usepackage{todonotes}
\usepackage{caption}

\SetMathAlphabet{\mathsf}{normal}{OT1}{lmss}{m}{n}
\SetMathAlphabet{\mathsf}{bold}{OT1}{lmss}{bx}{n}

\makeatletter
\def\nonumberfootnote{\xdef\@thefnmark{}\@footnotetext}			
\makeatother

\definecolor{colorred}{HTML}{B00000}
\definecolor{colorgreen}{HTML}{258300}
\definecolor{colorblue}{HTML}{2e32fa}
\definecolor{coloryellow}{HTML}{cbbb1a}
\hypersetup{colorlinks=true, linkcolor=colorred, citecolor=colorgreen, urlcolor=coloryellow}	

\numberwithin{equation}{section}

\newcommand{\rmd}{{\ensuremath{\mathrm{d}}}}
\newcommand{\rme}{{\ensuremath{\mathrm{e}}}}
\newcommand{\rmf}{{\ensuremath{\mathrm{f}}}}

\newcommand{\rmp}{{\ensuremath{\mathrm{p}}}}

\newcommand{\BdP}{{\ensuremath{\mathbf{P}}}}

\newcommand{\sfd}{{\ensuremath{\mathsf{d}}}}

\newcommand{\sfC}{{\ensuremath{\mathsf{C}}}}

\newcommand{\sfX}{{\ensuremath{\mathsf{X}}}}
\newcommand{\sfY}{{\ensuremath{\mathsf{Y}}}}

\newcommand{\scrA}{{\ensuremath{\mathscr{A}}}}
\newcommand{\scrB}{{\ensuremath{\mathscr{B}}}}
\newcommand{\scrC}{{\ensuremath{\mathscr{C}}}}

\newcommand{\scrF}{{\ensuremath{\mathscr{F}}}}

\newcommand{\N}{\boldsymbol{\mathrm{N}}}						
\newcommand{\R}{\boldsymbol{\mathrm{R}}}						
\renewcommand{\d}{\,\mathrm{d}}				

\DeclareMathOperator{\supp}{spt}			

\theoremstyle{definition}
\newtheorem{bump}{Bump}[section]

\theoremstyle{plain}
\newtheorem{theorem}[bump]{Theorem}
\newtheorem{proposition}[bump]{Proposition}
\newtheorem{definition}[bump]{Definition}
\newtheorem{lemma}[bump]{Lemma}
\newtheorem{corollary}[bump]{Corollary}

\theoremstyle{remark}
\newtheorem{remark}[bump]{Remark}
\newtheorem{example}[bump]{Example}

\crefname{theorem}{Theorem}{Theorems}
\crefname{proposition}{Proposition}{Propositions}
\crefname{definition}{Definition}{Definitions}
\crefname{lemma}{Lemma}{Lemmas}
\crefname{corollary}{Corollary}{Corollaries}
\crefname{hypothesis}{Hypothesis}{Hypotheses}
\crefname{remark}{Remark}{Remarks}
\crefname{example}{Example}{Examples}
\crefname{notation}{Notation}{Notations}

\renewenvironment{example}
  {\begin{oldexample}}
  {\hfill \scalebox{0.7}{$\blacksquare$}\end{oldexample}}

\renewenvironment{remark}
  {\begin{oldremark}}
  {\hfill \scalebox{0.7}{$\blacksquare$}\end{oldremark}}

\crefformat{section}{{§}#2#1#3}
\crefformat{subsection}{{§}#2#1#3}
\crefformat{subsubsection}{{§}#2#1#3}
\crefformat{appendix}{{§}#2#1#3}

\crefmultiformat{theorem}{Theorems #2#1#3}{ and #2#1#3}{, #2#1#3}{, and #2#1#3}
\crefmultiformat{proposition}{Propositions #2#1#3}{ and #2#1#3}{, #2#1#3}{, and #2#1#3}
\crefmultiformat{definition}{Definitions #2#1#3}{ and #2#1#3}{, #2#1#3}{, and #2#1#3}
\crefmultiformat{lemma}{Lemmas #2#1#3}{ and #2#1#3}{, #2#1#3}{, and #2#1#3}
\crefmultiformat{corollary}{Corollaries #2#1#3}{ and #2#1#3}{, #2#1#3}{, and #2#1#3}
\crefmultiformat{hypothesis}{Hypotheses #2#1#3}{ and #2#1#3}{, #2#1#3}{, and #2#1#3}
\crefmultiformat{remark}{Remarks #2#1#3}{ and #2#1#3}{, #2#1#3}{, and #2#1#3}
\crefmultiformat{example}{Examples #2#1#3}{ and #2#1#3}{, #2#1#3}{, and #2#1#3}
\crefmultiformat{notation}{Notations #2#1#3}{ and #2#1#3}{, #2#1#3}{, and #2#1#3}

\crefmultiformat{section}{{§§}#2#1#3}{ and #2#1#3}{, #2#1#3}{, and #2#1#3}
\crefmultiformat{subsection}{{§§}#2#1#3}{ and #2#1#3}{, #2#1#3}{, and #2#1#3}
\crefmultiformat{subsubsection}{{§§}#2#1#3}{ and #2#1#3}{, #2#1#3}{, and #2#1#3}

\crefrangeformat{equation}{#3\textcolor{black}{(}#1\textcolor{black}{)}#4 to #5\textcolor{black}{(}#2\textcolor{black}{)}#6}

\newcommand{\mms}{\mathsf{M}}				
\newcommand{\met}{\sfd}						

\newcommand{\meas}{\mathfrak{m}}				

\newcommand{\vol}{\mathrm{vol}}				
\newcommand{\Prob}{\mathscr{P}}		        

\newcommand{\pr}{\mathrm{pr}}				

\newcommand{\push}{\sharp}					

\newcommand{\cl}{\mathrm{cl}}				

\usepackage{blindtext}

\allowdisplaybreaks

\let\oldtocsection=\tocsection

\let\oldtocsubsection=\tocsubsection

\let\oldtocsubsubsection=\tocsubsubsection

\renewcommand{\tocsection}[2]{\hspace{0em}\oldtocsection{#1}{#2}}
\renewcommand{\tocsubsection}[2]{\hspace{1em}\oldtocsubsection{#1}{#2}}
\renewcommand{\tocsubsubsection}[2]{\hspace{2em}\oldtocsubsubsection{#1}{#2}}

\newcommand{\nocontentsline}[3]{}
\newcommand{\tocless}[2]{\bgroup\let\addcontentsline=\nocontentsline#1{#2}\egroup}

\renewcommand{\q}{\mathfrak{q}}

\newcommand{\PPP}{\BdP}

\allowdisplaybreaks

\makeatletter
\@namedef{subjclassname@2020}{\textup{2020} Mathematics Subject Classification}
\makeatother

\begin{document}

\title[Spacetime reconstruction by order and number]{Spacetime reconstruction by order and number}
\author{Mathias Braun}
\address{Institute of Mathematics, EPFL, 1015 Lausanne, Switzerland}
\email{\href{mailto:mathias.braun@epfl.ch}{mathias.braun@epfl.ch}}
\thanks{Financial support by the EPFL through a Bernoulli Instructorship is gratefully acknowledged. I sincerely thank Fay Dowker for bringing Bombelli's conjecture to my attention at the ``Kick-off Workshop `A new geometry for Einstein’s theory of relativity and beyond''' (University of Vienna, 2025). I thank its organizers Michael Kunzinger, Raquel Perales, Chiara Rigoni, Clemens Sämann, and Roland Steinbauer for the invitation to this inspiring event.}

\begin{abstract} We show that the random adjacency matrices induced by the chronological relations and i.i.d.~samples of two  space\-times coincide in law if and only if the spacetimes in question are smoothly isometric. A similar result holds for weighted spacetimes. In the smooth framework of our article, this relaxes the hypotheses of the recent Gromov reconstruction theorem in Lorentzian signature by Braun--Sämann from a.s.~isometry of the respective time separation functions to a.s.~order isometry. In a probabilistic way, our  result makes a key paradigm of  causal set theory rigorous: spacetime can be recovered by only knowing ``order'' and ``number'' of its points. It confirms a weak version  of Bombelli's conjecture; therefore, it  contributes to recent efforts of formalizing the Haupt\-ver\-mu\-tung (viz.~fundamental conjecture) of causal set theory.
\end{abstract}

\maketitle

\thispagestyle{empty}

\tableofcontents

\addtocontents{toc}{\protect\setcounter{tocdepth}{2}}

\clearpage

\section{Introduction}\label{Ch:Intro}

\subsection{From Hawking--King--Malament--McCarthy theory...} A well-known fact from Riemannian geometry established by Myers--Steenrod \cite{myers-steenrod1939} and Palais \cite{palais1957}  states a bijective map between two   Riemannian manifolds which preserves the respective distance functions is a smooth isometry. The converse holds for trivial reasons. Conceptually, this means the  metric structure of a Riemannian manifold already determines its  geometry.

The counterpart of this fact in Lorentzian geometry --- the mathematical foundation of Einstein's general relativity --- was shown by Hawking--King--McCarthy \cite{hawking-king-mccarthy1976}. We refer to \cref{Sub:Prel} for details about spacetimes and causality theory. Here and in the sequel, we assume  all spacetimes  have the same dimension $d\in\N$ no less than $3$. Hawking--King--McCarthy's result states a map between two strongly causal\footnote{A spacetime is called strongly causal if it has no almost closed causal curves.}   spacetimes is a smooth isometry if and only if it is bijective and preserves the respective time separation functions. In comparison with its Riemannian precedent, this result is quite  surprising. The time separation function is not a metric; for instance, it obeys the \emph{reverse} triangle inequality. Hence, it encodes not even topological information a priori. Still,  the time separation function of a spacetime suffices to recover its entire geometry.

In addition, Hawking--King--McCarthy \cite{hawking-king-mccarthy1976}  established a related theorem --- which is central for the spacetime re\-construction problem outlined further below and our work --- that was later sharpened by  Malament \cite{malament1977}. Recall that the positivity superlevel set of a time separation function  defines the chronological relation of a spacetime. The Hawking--King--Malament--McCarthy theorem shows affirmatively that chronology-preserving maps in fact determine the conformal class of a Lorentzian metric.

\begin{theorem}[Hawking--King--Malament--McCarthy's theorem  \cite{hawking-king-mccarthy1976,malament1977}]\label{Th:TheMALA} Assume  $(\mms,g)$ and $(\mms',g')$ constitute two distinguishing spacetimes. Then a map $\iota\colon\mms\to\mms'$ is a smooth conformal isometry if and only if it is a chronological isomorphism.
\end{theorem}

Here, a chronological isomorphism between spacetimes means a bijective chronology-preserving map, cf.~\cref{Def:Chronpres}.

As shown by Malament \cite{malament1977}, the above  hypothesis of distinction is sharp; the claim of \cref{Th:TheMALA} is wrong if merely future \emph{or} past distinction is assumed.

\begin{remark}[About bijectivity] In all results  above and  below, the possible requirement of bijectivity can be replaced by surjectivity. We will only deal with spacetimes for which maps that preserve a certain object or structure are injective, cf.~e.g.~\cref{Le:Inj}.
\end{remark}

\subsection{...to the causal set approach to quantum gravity}\label{Sub:Rel} The relevance of \cref{Th:TheMALA} in applications is best illustrated by the spacetime reconstruction problem from the perspective of causal set theory (CST). 

Broadly speaking, the spacetime reconstruction problem  lies at the heart of  quantum gravity. Some candidate theories --- including CST \cite{bombelli-lee-meyer-sorkin1987,surya2019,dowker-surya2024}, loop quantum gravity \cite{rovelli1998}, spin foam models \cite{perez2003}, or AdS/CFT \cite{maldacena1998} --- treat spacetime not as fundamentally given, but as a structure emerging from more primitive degrees of freedom, e.g.~causal order, algebraic relations, or quantum states. The spacetime reconstruction problem  thus bridges micro\-physical models and the macroscopic continuum limit. 

The CST perspective on quantum gravity  is ``distinguished by its logical simplicity and by the fact that it incorporates the assumption of an underlying spacetime discreteness organically and from the very beginning'', as popularized by Sorkin \cite{sorkin2024}. A slogan that describes its program  (cf.~e.g.~Surya \cite[p.~5]{surya2019}) widely attributed to  Sorkin  but already expressed in an unpublished CERN preprint of Myrheim \cite{myrheim1978} is
\begin{align}\label{Eq:ONG}
\textnormal{order} + \textnormal{number} = \textnormal{geometry}.
\end{align}
Early attempts to quantify physical implications of order (including efforts of Weyl and Lorentz) are almost as old as general rela\-ti\-vity, cf.~Robb \cite{robb1936}. The slogan \eqref{Eq:ONG}, which gained visibility in the early 2000s, was inspired by \cref{Th:TheMALA}. In a nutshell, 
\begin{itemize}
\item 9 of the 10 degrees of freedom in spacetime are predicted to come from \emph{order}, which fixes the conformal class of spacetime (cf.~Finkelstein \cite{finkelstein1969}), and 
\item the remaining degree should be fixed by \emph{number}  (that relates to volume), which determines the conformal factor.
\end{itemize}
This is best illustrated by the following folklore result and its proof. 

\begin{theorem}[Spacetime reconstruction from measure-preserving chronological isomorphisms]\label{Th:CI} Let $(\mms,g)$ and $(\mms',g')$ be   distinguishing spacetimes. Then a map $\iota\colon\mms\to\mms'$ is a smooth isometry if and only if it is a volume-preserving chronological isomorphism. 
\end{theorem}

\begin{proof} The ``only if'' implication is trivially true.

We now show the ``if'' implication. By \cref{Th:TheMALA}, $\iota$ is a smooth conformal isometry; in particular, the hypothesis about volume-preservation makes sense. This implies there is a smooth nowhere vanishing function $\smash{\Sigma\colon\mms\to\R}$ with $\smash{\iota^*g' = \Sigma^2\,g}$. Since $\smash{\iota^{-1}}$ is volume-preserving, standard transformation rules of volume measures, cf.~\cref{Le:Trafo},   imply 
\begin{align*}
\vol_g = \iota^{-1}_\push\vol_{g'} = \vol_{\iota^*g'} = \vol_{\Sigma^2\,g} = \Sigma^d\,\vol_g.
\end{align*}
Since $\smash{\vol_g}$ is fully supported and $\Sigma$ is smooth, this gives $\Sigma = 1$ everywhere on $\mms$. In turn, we obtain $\iota^*g' = g$, as desired.
\end{proof}

We are not aware of a rigorous mathematical result that captures \eqref{Eq:ONG} in the realm of CST, i.e.~in discrete language. Observe that \cref{Th:CI} (and \cref{Th:TheMALA}) is phrased in continuum, where counting the number of points has no  meaning. In the last decades, Poisson point processes (PPP, cf.~e.g.~Last--Penrose \cite{last-penrose2018}) have become a standard way to ``discretize''  spacetime stochastically, cf.~Myrheim \cite{myrheim1978}, Bombelli--Lee--Meyer--Sorkin \cite{bombelli-lee-meyer-sorkin1987}, Brightwell--Gregory \cite{brightwell-gregory1991}, and Bombelli \cite{bombelli2000}:  order is quantified by the order statistics from its samples and the  underlying  chronological relation, while the number component measured by cardinality is well-defined (as samples of PPPs are a.s.~locally finite) and consistent (as PPPs recover the volume measure in expectation). In terms of Poisson sprinklings, a rigorous statistical form of \eqref{Eq:ONG} was conjectured by Bombelli \cite{bombelli2000}. 

Our main result will rigorously reflect \eqref{Eq:ONG} and confirm a weak version of Bombelli's conjecture. We   elaborate on   this --- and possible connections to the most prominent open problem in CST, the \emph{Hauptvermutung} (viz.~fundamental conjecture) --- in \cref{Re:RelationBCHV}.

Despite the apparent lack of a rigorous version of  \eqref{Eq:ONG}, several works have recovered continuum quantities in spacetime \emph{approximately} from their discrete approximations by  infinite intensity limits of Poisson sprinklings.  Notable instances of such quantities are dimension and volume (Myrheim \cite{myrheim1978}), time separation function (Brightwell--Gregory \cite{brightwell-gregory1991}), scalar curvature (Benincasa--Dowker \cite{benincasa-dowker2010}), ``asymptotic silence''  (Eichhorn--Mizera--Surya \cite{eichhorn-mizera-surya2017}), and convex neighborhoods (Glaser--Surya \cite{glaser-surya2013}). In particular,  Brightwell--Gregory \cite{brightwell-gregory1991} demonstrate time separation functions are no  fundamental quantities in CST.

\subsection{Contributions}\label{Sub:Contri} 

\subsubsection{Main result} The main result  of our work is a spacetime reconstruction theorem. It is  a probabilistic and discrete analog of \cref{Th:CI} in the style of \eqref{Eq:ONG}.

Before stating our main \cref{Th:Main}, we introduce some preliminary modeling. Let $(\mms,g)$ be a spacetime of finite volume $\lambda > 0$. All stochastic objects will be modeled on the same probability space $(\Omega,\scrF,\PPP)$. Let $(X_i)_{i\in\N}$ be a sequence of i.i.d.~random variables $X_i\colon\Omega\to\mms$ with law $\meas$ under $\PPP$, where $\meas$ is the normalized volume measure
\begin{align}\label{Eq:NorVol}
\lambda\,\meas=\vol_g.
\end{align}
Given $k\in\N$, we  think of $\{X_1,\dots,X_k\}$ as the random support of a PPP conditioned on  $k$ elements. The order   statistics of $\{X_1,\dots,X_k\}$ with respect to the chronological relation $\ll$ from $(\mms,g)$ are described by the random adjacency matrix $\smash{\sfC^k(X_1,\dots,X_k)}$; here, the map $\smash{\sfC^k\colon\mms^k \to \{0,1\}^{k\times k}}$ is defined  by
\begin{align*}
\sfC^k(x_1,\dots,x_k)_{ij} := \begin{cases} 1 & \textnormal{if }x_i\ll x_j,\\
0 & \textnormal{otherwise}.
\end{cases}
\end{align*}
This turns $\{X_1,\dots,X_k\}$ into a random causal set. The law of $\smash{\sfC^k(X_1,\dots,X_k)}$ under $\PPP$ is
\begin{align}\label{Eq:nun}
\nu^k := \sfC^k_\push\meas^{\otimes k},
\end{align}
where $\sharp$ denotes the usual push-forward operation of measures.

Similar notions are adopted for a second spacetime $(\mms',g')$ with the same properties. Here and in the sequel, we adopt the following notational convention: if a quantity is tagged with a prime $'$, it is understood relative  to $(\mms',g')$.

\begin{theorem}[Probabilistic spacetime reconstruction, unweighted case]\label{Th:Main}  Let $(\mms,g)$ and $(\mms',g')$ be two causally continuous and future chronocomplete spacetimes of finite volume $\lambda > 0$. Then   the two spacetimes $(\mms,g)$ and $(\mms',g')$ are smoothly isometric if and only if for every $k\in\N$, the distributions  of the random adjacency matrices $\smash{\sfC^k(X_1,\dots,X_k)}$ and $\smash{\sfC'^k(X_1',\dots,X_k')}$ under $\PPP$ coincide, i.e.
\begin{align}\label{Eq:Hyyy}
\sfC^k(X_1,\dots,X_k)_\push\PPP = \sfC'^k(X_1',\dots,X_k')_\push\PPP.
\end{align}
\end{theorem}

In particular, the latter hypothesis is a probabilistic and discrete Lorentz invariant. 

In light of \eqref{Eq:ONG}, the random adjacency matrices $\smash{\sfC^k(X_1,\dots,X_k)}$ and $\smash{\sfC'^k(X_1',\dots,X_k')}$ in \cref{Th:Main} determine the statistical ``order'' of the samples in question, whereas the dimension $k\in\N$ corresponds to their ``number''.

We briefly comment on the hypotheses of \cref{Th:Main}. Causal continuity, a classical causality condition defined by Hawking--Sachs \cite{hawking-sachs1974}, means distinction and reflectivity, cf.~\cref{Def:Causalcontinuity,Def:Distinction,Def:Reflectivity}. In view of \cref{Th:TheMALA}, distinction is natural. The notion of chronocompleteness is new, cf.~\cref{Def:RelF}. Roughly speaking, it asserts that every sequence which increases and is upper bounded along the chronological relation must converge. This is a weak variant of future complete metric measure spacetimes recently defined by Beran et al.~\cite{beran-braun-calisti-gigli-mccann-ohanyan-rott-samann+-} and further elaborated on by Gigli \cite{gigli+}. Such a ``completion procedure'' was already anticipated e.g.~by Bombelli \cite{bombelli2000}; completions of (strict) partial orders had already been known before in theoretical computer science, cf.~e.g.~Gierz et al.~\cite{gierz-hofmann-keimel-lawson-mislove-scott1980}. Global  hyperbolicity implies all these properties, cf.~\cref{Le:GHtoFC}.

The ``only if'' implication from \cref{Th:Main} is trivially true. An ingredient for the proof of the ``if'' implication of  \cref{Th:Main} that our work  pioneers --- which we think is of independent interest --- is an extension result for  chronology-preserving maps, cf.~\cref{Th:Ext}. In a nutshell, future chronocompleteness yields existence of an extension by a natural limit procedure; reflectivity yields uniqueness of this (one-sided) limit. Helpful tools for the proof of \cref{Th:Ext} will be Minguzzi's $D$-relations  \cite{minguzzi2008-ladder-i}; for causally continuous space\-times, they coincide with Sorkin--Woolgar's $K$-relation \cite{sorkin-woolgar1996}, as proven by Minguzzi \cite{minguzzi2008-ladder-i}.

\cref{Th:Main} is in the spirit of  Gromov's  celebrated  reconstruction theorem  for metric measure spaces \cite{gromov1981} and its recent adaptation to Minguzzi--Suhr's bounded Lorentzian metric spaces \cite{minguzzi-suhr2022} by Braun--Sämann \cite{braun-samann+}. The latter does not apply here, as it makes  hypotheses on the  time separation functions in question, which are not invariant under conformal changes, unlike our setting. Nevertheless, using our extension \cref{Th:Ext} we will use Vershik's beautiful proof strategy presented in Gromov's book \cite{gromov1981} and its Lorentzian adaptation  by Braun--Sämann \cite{braun-samann+} to show \cref{Th:Main}.

\subsubsection{Weighted extension of \cref{Th:Main}} Our second main result is a weighted version of \cref{Th:Main}. In this generality, we lose rigidity encoded by the reference measures --- cf.~the proof of \cref{Th:CI} --- which clarifies we cannot expect isometry of the original spacetimes. However, we still get isometry of explicit conformal changes, cf.~\cref{Re:IsoCon}.

We adopt the modeling for \cref{Th:Main} except that the normalized volume measure of $(\mms,g)$ --- and analogously for the other spacetime --- is replaced by the normalized weighted volume measure with prescribed continuous potential  $V\colon\mms\to\R$:
\begin{align}\label{Eq:NorVolII}
\lambda\,\meas = \rme^V\,\vol_g.
\end{align}

\begin{theorem}[Probabilistic spacetime reconstruction, weighted case]\label{Th:Weighted} We let $(\mms,g)$ and $(\mms',g')$ be two causally continuous and future chronocomplete spacetimes. Given  two continuous functions $V\colon\mms\to\R$ and $V'\colon\mms'\to \R$, we endow $(\mms,g)$ and $(\mms',g')$ with the weighted volume measures $\smash{\rme^V\,\vol_g}$ and $\smash{\rme^{V'}\,\vol_{g'}}$, respectively. Assume the latter have finite mass $\lambda > 0$. Then the measured spacetimes $\smash{(\mms,g,\rme^V\,\vol_g)}$ and $\smash{(\mms',g',\rme^{V'}\,\vol_{g'})}$ are smoothly con\-formally isometric  through a measure-preserving map $\iota\colon\mms\to\mms'$ if and only if for every $k\in\N$, the laws of the random adjacency  matrices $\smash{\sfC^k(X_1,\dots,X_k)}$ and $\smash{\sfC'^k(X_1',\dots,X_k')}$ under $\PPP$ coincide.

In either case, the conformal factor $\Sigma\colon \mms\to \R$ certifying $\iota^*g' = \Sigma^2\,g$  obeys
\begin{align}\label{Eq:Rel}
\Sigma^d = \rme^V\,\rme^{-V'\circ\iota}.
\end{align}
\end{theorem}

Notably, this result does not require the spacetimes in question to be bounded.

\begin{remark}[About isometry]\label{Re:IsoCon} The smooth and measure-preserving conformal isometry from  \cref{Th:Weighted} can equivalently be stated as smooth isometry of the two conformally changed  spacetimes $\smash{(\mms,\rme^{2V/d}\,g)}$ and $\smash{(\mms',\rme^{2V'/d}\,g')}$.
\end{remark}

\subsection{Outlook} \cref{Th:Main} is a weak form of Bombelli's conjecture \cite{bombelli2000}. It would be a natural yet challenging extension of our main result. We detail this in \cref{Re:RelationBCHV}.

Let us comment on another possible generalization of  \cref{Th:Main,Th:Weighted}. These   involve only the chronological orders and reference measures of the spacetimes in question.   On the other hand, in recent years intense research activities were devoted to Lorentzian geometry by means  of metric geometry. We refer to the reviews of Cavalletti--Mondino \cite{cavalletti-mondino2022-review}, Sämann \cite{samann2024+}, McCann \cite{mccann+}, and Braun \cite{braun2025}. An abstract  ``converse'' Hawking--King--McCarthy theorem was recently proven  by Beran et al.~\cite{beran-braun-calisti-gigli-mccann-ohanyan-rott-samann+-}. These evidences motivate the question if \cref{Th:Main,Th:Weighted} hold in a synthetic way  (where, of course, the smooth conclusions should be replaced by weaker concepts); an extension to spacetimes with nonsmooth Lorentz\-ian metrics would  be equally interesting. In general, we only expect an analog of \cref{Th:Weighted}. As  \eqref{Eq:Rel}  shows, an extension of \cref{Th:Main} would necessarily include a rigid hypothesis on the reference measures. Yet, \cref{Re:IsoCon} seems to suggest this synthetification might be more natural in Kunzinger--Sämann's Lorentzian pre-length spaces \cite{kunzinger-samann2018} or Minguzzi--Suhr's bounded  Lorentzian metric spaces \cite{minguzzi-suhr2022} (and their unbounded extension  with Bykov \cite{bykov-minguzzi-suhr2025}) instead of Nachbin's abstract topological ordered spaces \cite{nachbin1965}. Either should be endowed with reference measures; a natural candidate for  a synthetic counterpart of \cref{Th:Main} is McCann--Sämann's Lorentz\-ian Hausdorff measure \cite{mccann-samann2022}.

\subsection{Organization} In  \cref{Sub:Prel}, we introduce  basic notions of probability theory, spacetime geometry, and causality theory. Moreover, we recall the fundamental results of  Hawking--King--McCarthy \cite{hawking-king-mccarthy1976} and Malament \cite{malament1977}.  In \cref{Sub:Prooof}, we pioneer our extension theorems for chronology-preserving maps and use these insights to prove \cref{Th:Main,Th:Weighted}.

\section{Preliminaries}\label{Sub:Prel}

\subsection{Probabilistic concepts} In the following, we briefly recall some basic concepts from probability theory that will be relevant for our discussion.

We first fix some terminology. A probability space $(\Omega,\scrF,\PPP)$ will always consist  of a measurable space  $(\Omega,\scrF)$ and a probability measure $\PPP$ (also called distribution) on it. By a random variable, we mean  a measurable map $X$ defined on $\Omega$ with values in a specified measurable space. The push-forward $X_\push\PPP$  is called \emph{law} of $X$ under $\PPP$. 

\subsubsection{Standard Borel spaces} A \emph{standard Borel space} is a measurable space $(\sfX,\scrA)$ such that there exists a Polish topology on $\sfX$ whose induced Borel $\sigma$-algebra coincides with $\scrA$; cf.~Kechris \cite{kechris1995}*{§12} for details. In particular, measurability with respect to standard Borel spaces is equivalent to Borel measurability. We let $\Prob(\sfX)$ designate  the space of all probability measures on such a space, endowed with the narrow topology (induced by convergence against bounded and continuous functions on $\sfX$  with values in $\R$); cf.~e.g. Ambrosio--Gigli--Savaré \cite{ambrosio-gigli-savare2008}*{§5.1} for details.

Given any standard Borel space $(\sfX,\scrA)$, the \emph{support} of an element $\nu\in \Prob(\sfX)$, denoted $\supp\nu$, is the smallest closed set with $\nu$-measure one. It exists, cf.~Kechris \cite{kechris1995}*{Ex.~17.38}.

\subsubsection{Kolmogorov's extension theorem}\label{Sub:Kolmo} The classical fact we recall now will only be needed and defined for countable families. We will follow the presentation from Kechris \cite{kechris1995}*{Ex.~17.16} and refer to  Tao \cite{tao2011}*{Thm.~2.4.3} for details.

Let $(\sfX_k,\scrA_k)_{k\in\N}$ constitute a sequence of measurable spaces. Given any $k\in\N$, let $\smash{f_{k+1}\colon\sfX_{k+1} \to \sfX_k}$ be a surjective measurable map. We define
\begin{align*}
\lim_{\leftarrow} \sfX_\bullet := \Big\lbrace (x_k)_{k\in\N} \in\prod_{k\in\N} \sfX_k : f_{k+1}(x_{k+1}) = x_k  \textnormal{ for every }k \in\N\Big\rbrace.
\end{align*}
As usual, let $\smash{\pr_k\colon\lim_{\leftarrow}\sfX_\bullet\to \sfX_k}$ designate the projection onto the $k$-th coordinate of its argument, where $k\in\N$.  
We endow the set $\lim_{\leftarrow}\sfX_\bullet$ with the $\sigma$-algebra
\begin{align*}
\lim_{\leftarrow} \scrA_\bullet := \sigma\Big[\!\bigcup_{k\in\N} \pr_k^{-1}(\scrA_k)\Big].
\end{align*}
The measurable space $(\lim_\leftarrow\sfX_\bullet,\lim_\leftarrow\scrA_\bullet)$ is called \emph{projective limit} of $(\sfX_k,\scrA_k)_{k\in\N}$ with respect to $(f_k)_{k\in\N}$. It is a standard Borel space if $(\sfX_k,\scrA_k)$ is so for every $k\in\N$.

We  call a sequence $(\nu_k)_{k\in\N}$ of probability measures $\nu_k$ on $(\sfX_k,\scrA_k)$ \emph{projective} with respect to $(f_k)_{k\in\N}$ if $(f_{k+1})_\push\nu_{k+1} = \nu_k$ for every $k\in\N$.

\begin{theorem}[Kolmogorov's extension theorem]\label{Th:Kolmogorov} Assume $(\sfX_k,\scrA_k)_{k\in\N}$ is a sequence of standard Borel spaces. Let $(f_k)_{k\in\N}$ be as above. Moreover, let $(\nu_k)_{k\in\N}$ be a projective family of probability measures $\nu_k$ on $(\sfX_k,\scrA_k)$ with respect to $(f_k)_{k\in\N}$. Then there is a unique probability measure $\lim_\leftarrow\nu_\bullet$ on $(\lim_\leftarrow\sfX_\bullet,\lim_\leftarrow\scrA_\bullet)$  such that for every $k\in\N$,
\begin{align*}
(\pr_k)_\push\big[\!\lim_\leftarrow\nu_\bullet\big] = \nu_k.
\end{align*}
\end{theorem}

\subsection{Spacetime geometry} Now we recall some basic notions of Lorentzian geometry. We refer to the classical textbook of  Beem--Ehrlich--Easley  \cite{beem-ehrlich-easley1996} for details.

\subsubsection{Spacetimes} Throughout, every smooth topological manifold $\mms$ is assumed to be Hausdorff, second-countable, connected, without boundary, and with dimension at least two. A \emph{Lorentzian metric} on $\mms$ is a smooth section $g$ of $(T^*)^{\otimes 2}\mms$ of constant signature $+,-,\dots,-$. Given  any $v\in T\mms$, we write $\smash{\vert v\vert := \sqrt{g(v,v)}}$ if the radicand is nonnegative.

Note that if $\Sigma\colon \mms \to\R$ is a smooth nowhere vanishing function, then the pointwise multiplication $\Sigma^2\,g$ is again a Lorentzian metric on $\mms$.

In addition, if $\iota\colon\mms\to\mms'$ forms a diffeomorphism between two smooth topological manifolds $\mms$ and $\mms'$ and $g'$ constitutes a Lorentzian metric on $\mms'$, the \emph{pull-back} $\iota^*g'$ is a Lorentzian metric on $\mms$  defined as follows for every $v,w\in T_x\mms$, where $x\in\mms$:
\begin{align*}
(\iota^*g')_x(v,w) := g'_{\iota(x)}(\rmd\iota_{x}\,v,\rmd\iota_x\,w).
\end{align*}

A tangent vector $v\in T\mms\setminus\{0\}$ will be called  \emph{timelike} if $g(v,v) > 0$, \emph{lightlike} provided  $g(v,v) = 0$,  \emph{causal} if it is timelike or lightlike, and  \emph{spacelike} if $g(v,v) < 0$. The zero vector is spacelike by convention. These adjectives define the so-called causal character of $v$. Correspondingly, a smooth vector field $X$ on $\mms$ is called timelike, light\-like, causal, or spacelike provided the evaluation $X_x$ has the respective causal character for every $x\in\mms$. In particular, this sets up the causal character of a  smooth curve $\gamma\colon I\to \mms$ defined on an interval $I\subset\R$ in terms of the causal character of the smooth vector field $\dot\gamma$ along $\gamma$.

A \emph{time orientation} refers to the choice of a continuous timelike vector field $T$ on $\mms$. We  call the triple $(\mms,g,T)$ --- or simply the pair $(\mms,g)$, where the time orientation $T$ is regarded as implicitly understood --- \emph{spacetime}. We will term $v\in T_x\mms$ \emph{future-directed} if $g(v,X_x) >0$, where $x\in\mms$. In an analogous  way, future-directed smooth vector fields and smooth curves are set up, respectively. Unless explicitly stated otherwise, all relevant objects will be assumed to be future-directed.

\subsubsection{Reference measures} Let $\vol_g$ be the usual volume measure on $\mms$ induced by $g$. We  also occasionally consider the weighted volume measure $\smash{\rme^V\,\vol_g}$, where $V\colon\mms\to\R$ is a continuous function. The subsequent result  summarizes standard properties of such measures, which follow from the Euclidean change of variables formula. 

\begin{lemma}[Transformation rules]\label{Le:Trafo} Let $(\mms,g)$ and $(\mms',g')$ be two given spacetimes. Let $V'\colon \mms'\to\R$ be continuous. Then the following identities hold as Borel measures on $\mms$.
\begin{enumerate}[label=\textnormal{(\roman*)}]
\item \textnormal{\textbf{Push-forwards vs.~pull-backs.}} If  $\smash{\iota\colon\mms\to\mms'}$ is a diffeomorphism,
\begin{align*}
\iota^{-1}_\push\big[\rme^{V'}\,\vol_{g'}\big] = \rme^{V'\circ\iota}\,\vol_{\iota^*g'}.
\end{align*}
\item \textnormal{\textbf{Conformal changes.}}  If $\Sigma\colon \mms\to\R$ is a smooth nowhere vanishing function,
\begin{align*}
\vol_{\Sigma^2\,g}=\Sigma^d\,\vol_g.
\end{align*}
\end{enumerate}
\end{lemma}

\subsection{Causality theory}\label{Sub:Causthry} For details about the material about Lorentzian causality theory collected now, we refer to the review of Minguzzi \cite{minguzzi2019-causality}.

\subsubsection{Chronology and causality} First, we recall the two standard relations $\ll$ and $\leq$ of \emph{chronology} and \emph{causality}, respectively.  Given  $x,y\in \mms$, we write 
\begin{itemize}
\item $x\ll y$ if there exists a smooth timelike curve $\gamma\colon[0,1]\to\mms$ from $x$ to $y$ and
\item $x\leq y$ if there exists a smooth causal  curve $\gamma\colon[0,1]\to \mms$ connecting $x$ to $y$.
\end{itemize}
By these relations, one can define a merely continuous curve through $\mms$ to be timelike or causal, respectively, in the evident manner. Clearly, $\ll$ is contained in $\leq$. As well-known, both relations are transitive, $\leq$ is reflexive, and the push-up property holds: the composed relations $\ll \circ \leq$ and $\leq \circ \ll$ are both included in $\ll$; more precisely, for every $x,y,z\in\mms$ we have $x\ll z$ provided $x\leq y\ll z$ or $x\ll y\leq z$.  Given any  $x,y\in\mms$, we define the \emph{chronological future} of $x$  by $\smash{I^+(x) := \{z\in\mms : x\ll z\}}$ and  the \emph{chronological past} of $y$ by  $\smash{I^-(y):= \{z\in\mms : z\ll y\}}$.  
Analogous notions and notations will be adopted for the causality relation by replacing each occurrence of ``$I$'' and ``$\ll$'' by ``$J$'' and ``$\leq$'', respectively. Note that chronological futures and pasts are always open, whereas causal futures and pasts are in general not closed.

\subsubsection{Minguzzi's $D$-relations} There are many other relations one can define similarly. For our purposes, Minguzzi's $D$\footnote{This nomenclature comes from  the strong relation to the distinction property, cf.~\cite{minguzzi2019-causality}*{§4.3.2}.}-relations  \cite{minguzzi2008-ladder-i} will be useful. Given $x,y \in\mms$, we write
\begin{itemize}
\item $\smash{x\leq_\rmf y}$ if $\smash{y\in \cl\,J^+(x)}$ and
\item $\smash{x\leq_\rmp y}$ if $\smash{x\in \cl\,J^-(y)}$.
\end{itemize}
We have $\smash{x\leq_\rmf y}$ if and only if $\smash{I^+(x)\supset I^+(y)}$ as well as $\smash{x\leq_\rmp y}$ if and only if $\smash{I^-(x)\subset I^-(y)}$, cf. Minguzzi \cite{minguzzi2019-causality}*{p.~105}. Dowker--Garcia--Surya \cite{dowker-garcia-surya2000}*{Cl.~1} and Minguzzi \cite{minguzzi2008-ladder-i}*{Thm.~3.3} prove both relations are reflexive and transitive; \cref{Th:Antisy}  will  yield  a sufficient condition for its antisymmetry we use later. Lastly, as shown by Minguzzi \cite{minguzzi2008-weak-distinction}*{Lem. 2.7}, the push-up property holds. In fact, $\smash{\leq_\rmf}$ and $\smash{\leq_\rmp}$ are the largest relations --- with respect to the respective time orientation --- with this property, cf.~Minguzzi \cite{minguzzi2008-weak-distinction}*{Lem.~2.8}.

\begin{lemma}[Push-up property of Minguzzi's $D$-relations \cite{minguzzi2008-weak-distinction}]\label{Le:Push-up} The two composed relations $\smash{\ll \circ\leq_\rmf}$ and $\smash{\leq_\rmp \circ\ll}$ are contained in $\ll$. More precisely, for every $x,y,z\in\mms$ we have $x \ll z$ provided $\smash{x\leq_\rmf y \ll z}$ or $\smash{x\ll y\leq_\rmp z}$.
\end{lemma}

\subsubsection{Some causality conditions} Many different conditions of varying strengths can be imposed on these relations and their extensions. This leads to the vast field of Lorentzian causality theory, which is systematically  described in a unified way by Minguzzi \cite{minguzzi2019-causality}, to where we refer for details. We focus only on the conditions needed below. 

\begin{definition}[Distinction]\label{Def:Distinction} A spacetime $(\mms,g)$ will be called 
\begin{enumerate}[label=\textnormal{\alph*.}]
\item \emph{future distinguishing} if for every $x,y\in\mms$, we have $\smash{I^+(x)\neq I^+(y)}$ if $x\neq y$,
\item \emph{past distinguishing} if for every $x,y\in\mms$, we have $\smash{I^-(x)\neq I^-(y)}$ if $x\neq y$, and
\item \emph{distinguishing} if it is simultaneously future and past distinguishing.
\end{enumerate}
\end{definition}

That is, a spacetime $(\mms,g)$ is distinguishing if and only if for every $x,y\in\mms$, we have $x=y$ if and only if $\smash{I^+(x) = I^+(y)}$ or $\smash{I^-(x) = I^-(y)}$.

\begin{definition}[Reflectivity]\label{Def:Reflectivity} A spacetime $(\mms,g)$ is termed
\begin{enumerate}[label=\textnormal{\alph*.}]
\item \emph{future reflecting} if for every $x,y\in\mms$,  $\smash{I^-(x)\subset I^-(y)}$ implies $\smash{I^+(y) \subset I^+(x)}$,
\item \emph{past reflecting} if for every $x,y\in\mms$, $\smash{I^+(x)\subset I^+(y)}$ implies $\smash{I^-(y)\subset I^-(x)}$, and
\item \emph{reflecting} if it is simultaneously future and past reflecting.
\end{enumerate}
\end{definition}

\begin{definition}[Causal continuity]\label{Def:Causalcontinuity} A distinguishing and reflecting spacetime is called \emph{causally continuous}.
\end{definition}

The notion of reflectivity was introduced by Kronheimer--Penrose \cite{kronheimer-penrose1967}. Distinction and causal continuity stem from Hawking--Sachs \cite{hawking-sachs1974}. An equivalent yet a priori weaker notion for causal continuity was proposed by Minguzzi \cite{minguzzi2008-weak-distinction}. There are many equivalent characterizations of distinction and reflectivity; see Minguzzi \cite{minguzzi2019-causality}*{§§4.1, 4.3} for details.

Let us collect some basic facts. On every dist\-inguishing space\-time,  the chronological  relation $\ll$ is irreflexive, hence a strict partial order, cf.~Minguzzi  \cite{minguzzi2019-causality}*{Tab.~1}.

\begin{proposition}[Antisymmetry of Minguzzi's $D$-relations  {\cite{minguzzi2008-ladder-i}*{Thm.~3.3}}]\label{Th:Antisy} The $D$-relations $\smash{\leq_\rmf}$ and $\smash{\leq_\rmp}$ on a distinguishing spacetime are  antisymmetric, hence partial orders.
\end{proposition}

In view of the following definition,  we  call a sequence $(x_n)_{n\in\N}$ in a given spacetime $(\mms,g)$ \emph{chronologically increasing} if $x_n\ll x_{n+1}$ for every $n\in\N$, respectively. Then the \emph{chronological  decreasingness} of $(x_n)_{n\in\N}$ is defined analogously.

\begin{definition}[Chronocompleteness]\label{Def:RelF} We call a spacetime $(\mms,g)$ 
\begin{enumerate}[label=\textnormal{\alph*.}]
\item \emph{future chronocomplete} if every chronologically increasing sequence $(x_n)_{n\in\N}$ in $\mms$ such that there is $\smash{x^+\in \mms}$  with $\smash{x_n\ll x^+}$ for every $n\in\N$ converges,
\item \emph{past chronocomplete} if every chronologically decreasing sequence $(x_n)_{n\in\N}$ in $\mms$ such that there is $\smash{x^-\in \mms}$ with $\smash{x^-\ll x_n}$ for every $n\in\N$ converges, and
\item \emph{chronocomplete} if it is future and past chronocomplete.
\end{enumerate}
\end{definition}

The previous new notion does not appear in Minguzzi's review \cite{minguzzi2019-causality}. It constitutes a relaxation of the notion of (forward or backward) complete spacetimes which originates in our the recent work with Beran et al.~\cite{beran-braun-calisti-gigli-mccann-ohanyan-rott-samann+-}*{Def.~2.1} and was subsequently elaborated on by Gigli \cite{gigli+} in the context  of general partially ordered sets. Their notion mimics the notion of completeness of metric spaces in Lorentzian signature (much in the spirit of the two different yet equivalent concepts of Cauchy and Dedekind completeness of $\R$). However, note carefully we only ask for convergence of sequences that increase chronologically (as opposed to causally). In addition, we require the ``upper and lower bounds''  $\smash{x^+}$ and $\smash{x^-}$ in \cref{Def:RelF} to belong to the spacetime.

A large class of examples of distinguishing and chronocomplete spacetimes are globally hyperbolic ones. Its subsequent formulation follows the Bernal--Sánchez characterization \cite{bernal-sanchez2007} of the classical notion of global hyperbolicity, cf.~e.g.~Hawking--Ellis \cite{hawking-ellis1973}.

\begin{definition}[Global hyperbolicity] A spacetime $(\mms,g)$ is called \emph{globally hyperbolic} if 
\begin{enumerate}[label=\textnormal{\alph*.}]
\item the causality relation $\leq$ is antisymmetric and 
\item the causal diamond $\smash{I^+(x)\cap I^-(y)}$ is compact for every points $x,y\in\mms$.
\end{enumerate}
\end{definition}

\begin{lemma}[Implications of global hyperbolicity]\label{Le:GHtoFC} Every globally hyperbolic space\-time   is causally continuous and chronocomplete.
\end{lemma}

Indeed, global hyperbolicity is well-known to be  stronger than causal continuity, cf. Minguzzi \cite{minguzzi2019-causality}*{Fig.~20}. As noted by Beran et al.~\cite{beran-braun-calisti-gigli-mccann-ohanyan-rott-samann+-}*{Rem.~2.10}, every globally hyperbolic spacetime is forward and backward complete, hence chronocomplete.

\subsection{Hawking--King--Malament--McCarthy theory}\label{Sub:HKMM} We now recall some fundamental  contributions of Malament \cite{malament1977} based on previous work of Hawking--King--McCarthy \cite{hawking-king-mccarthy1976}. In a nutshell, they gave criteria on the spacetimes in question under which a map which preserves distances \cite{hawking-king-mccarthy1976} or orders \cite{malament1977} upgrades to a smooth isometry. Among these classes of maps, we will focus on  the two extreme ``regularities'', i.e.~metric- (possibly modulo conformal change) and order-preserving ones. 

\begin{definition}[Isometry]\label{Def:Isometry} We call a bijective map $\iota\colon\mms\to \mms'$
\begin{enumerate}[label=\textnormal{\alph*.}]
\item \emph{smooth isometry} if $\iota$ is a diffeomorphism which satisfies $\smash{\iota^*g'=g}$,
\item \emph{smooth conformal isometry} if $\iota$ is a diffeomorphism and there is a smooth  nowhere vanishing function $\Sigma\colon \mms'\to\R$ such that $\smash{\iota^*g' = \Sigma^2\,g}$.
\end{enumerate}
\end{definition}


\begin{definition}[Chronology-preservation]\label{Def:Chronpres} We  call a map $\iota\colon\mms\to\mms'$  
\begin{enumerate}[label=\textnormal{\alph*.}]
\item  \emph{chronology-preserving} if $x\ll y$ if and only if $\iota(x) \ll'\iota(y)$ for every $x,y\in\mms$,
\item  \emph{chronological isomorphism} if it is chronology-preserving and bijective.
\end{enumerate}
\end{definition}

A  chronological isomorphism is nothing but an isomorphism of the strict partial orders given by the respective chronological relations.

All the above properties are clearly stable under composition. 

\begin{theorem}[Malament's theorem \cite{malament1977}*{Thm.~1}]\label{Th:Malament} Assume that $(\mms,g)$ and $(\mms',g')$ are two distinguishing spacetimes. Let $\iota\colon \mms\to \mms'$ be a chronological isomorphism. Then $\iota$ constitutes a homeomorphism and in fact a smooth conformal isometry. 
\end{theorem}

As Malament noted \cite{malament1977}*{Fn.~4}, the  above  property of smooth conformal isometry from \cref{Th:Malament} is in fact a direct consequence of homeomorphy plus a minor modification of the Hawking--King--McCarthy theorem \cite{hawking-king-mccarthy1976}*{Thm.~5}. 

The following is a simple consequence of distinction we will use below.

\begin{lemma}[Injectivity]\label{Le:Inj} Assume $(\mms,g)$ and $(\mms',g')$ are distinguishing spacetimes. Let $\iota\colon \mms\to \mms'$ be chronology-preserving. Then $\iota$ is injective.
\end{lemma}

\begin{proof} Suppose $x,y\in\mms$ are distinct. Since  $(\mms,g)$ is distinguishing, one of the  following four sets is nonempty: $\smash{I^+(x) \setminus I^+(y)}$, $\smash{I^+(y) \setminus I^+(x)}$, $\smash{I^-(x) \setminus I^-(y)}$, or $\smash{I^-(y) \setminus I^-(x)}$. We assume nonemptiness of the first; the proof is analogous in the other situations.  Thus, let $\smash{z\in I^+(x) \setminus I^+(y)}$. Since $\iota$ is chronology-preserving, we have $\iota(x) \ll \iota(z)$ yet $\smash{\iota(y) \not\ll\iota(z)}$. This implies $\smash{I'^+(\iota(x))\setminus I'^+(\iota(y))}$ is nonempty. As $(\mms',g')$ is distinguishing, we conclude the images $\iota(x),\iota(y)\in\mms'$ must be distinct.
\end{proof}

\section{Proofs of \cref{Th:Main,Th:Weighted}}\label{Sub:Prooof}

\subsection{Extension of chronology-preserving maps} In the proof of \cref{Th:Main}, we will build countable dense sets from the two spacetimes in question. By assigning the points in an appropriate way, the induced map will be chronology-preserving. Developing the mechanism to extend this map beyond the  dense set in a unique chronology-preserving way is the objective of this part. As we believe the results are of independent interest, we have separated them from the actual proof of \cref{Th:Main}.

We flag a crucial property we assume throughout below: density of the image of the initial map to be extended. Later, this will be given by construction, cf.~\cref{Le:Dens}.

For the proofs to follow, recall Minguzzi's $D$-relations $\smash{\leq_\rmf}$ and $\smash{\leq_\rmp}$ \cite{minguzzi2008-ladder-i}  from  \cref{Sub:Causthry}.

\begin{proposition}[Uniqueness of extension]\label{Pr:Uniq} Let $(\mms,g)$ and $(\mms',g')$ be two spacetimes, where $(\mms',g')$ is distinguishing, past reflecting, and future  chronocomplete. Given any  dense set $D\subset\mms$, suppose $\iota\colon D\to \mms'$ is chronology-preserving with dense image. Then every  extension $\smash{\tilde{\iota}\colon\mms\to\mms'}$ of $\iota$ has the following property  for every $x\in\mms$ and every chronologically increasing sequence $(x_i)_{i\in\N}$ in $D$ converging to $x$:
\begin{align}\label{Eq:Equ!}
\tilde{\iota}(x) = \lim_{i\to\infty} \iota(x_i).
\end{align}

In particular, there exists at most one extension of $\iota$  to all of $\mms$.
\end{proposition}

\begin{proof} The last statement about uniqueness is clear from \eqref{Eq:Equ!}.

Now let $\smash{\tilde{\iota}\colon\mms\to\mms'}$ be an  extension of $\iota$ to all of $\mms$. Let $(x_i)_{i\in\N}$ be a chronologically increasing sequence in $D$ converging to $x$. By chronology-preservation of the extension, $(\iota(x_i))_{i\in\N}$ is a chronologically increasing sequence obeying $\smash{\iota(x_i)\ll'\tilde{\iota}(x)}$ for every $i\in\N$. Therefore, future chronocompleteness of $(\mms',g')$ yields the existence of a limit $y'\in\mms'$ of  $(\iota(x_i))_{i\in\N}$. We claim that $\smash{y' = \tilde{\iota}(x)}$. The inclusion $\smash{\iota(x_i) \ll' \tilde{\iota}(x)}$ for every $i\in\N$ gives that $\smash{y' \leq_\rmp' \tilde{\iota}(x)}$. On the other hand, let $\gamma\colon (0,1)\to\mms$ be a smooth timelike curve with midpoint $x$. \cref{Le:PresTL} below shows $\smash{\tilde{\iota}\circ\gamma}$ is a continuous timelike curve through  $\smash{\tilde{\iota}(x)}$. Given any $\varepsilon\in (0,1/2)$, we then have $\gamma_{1/2-\varepsilon} \ll x_i$ for every large enough $i\in\N$, and thus $\smash{\tilde{\iota}(\gamma_{1/2-\varepsilon}) \ll' \iota(x_i)}$. This yields $\smash{y' \in \cl\,J'^+(\gamma_{1/2-\varepsilon})}$ and therefore $\smash{\gamma_{1/2-\varepsilon} \in \cl\,J'^-(y')}$ using   past reflectivity of $(\mms',g')$. By closedness, continuity of $\iota\circ\gamma$, and the arbitrariness of $\varepsilon$, this yields $\smash{\tilde{\iota}(x) \in\cl\,J'^-(y')}$ and thus  $\smash{\tilde{\iota}(x) \leq_\rmp' y'}$. Since the relation $\smash{\leq_\rmp'}$ is antisymmetric by the distinction property, the claim is proven.
\end{proof}

The subsequent result employed above is a variant of Malament \cite{malament1977}*{Lem.~3}. Instead of bijectivity, we assume the map in question to have dense image; moreover, we only hypothesize the distinction property for the target space.  For the reader's convenience, we include a proof, although it is almost the same as in \cite{malament1977}.

\begin{lemma}[Preservation of continuous timelike curves]\label{Le:PresTL} Let $(\mms,g)$ and $(\mms',g')$ be two   spacetimes with $(\mms',g')$ distinguishing. Suppose $\iota\colon\mms\to\mms'$ is a chronology-preserving map with dense image. Then $\iota$ preserves continuous timelike curves.
\end{lemma}

\begin{proof}  To relax the notation, we only work with curves defined on $(0,1)$. Suppose now $\gamma\colon(0,1)\to\mms$ is a continuous timelike curve. Given any $t_0\in (0,1)$, let  $O'\subset\mms'$ form an open convex neighborhood of $\iota(\gamma_{t_0})$. It suffices to show the existence of an  open set $I\subset (0,1)$ containing $t_0$ with  $I\subset (\iota\circ\gamma)^{-1}(O')$, viz.~$\iota(\gamma_t)\in O'$ for every $t\in I$. 

As $(\mms',g')$ is future distinguishing, there exists an open set $U'\subset O'$ containing $\iota(\gamma_{t_0})$ such that every continuous timelike curve starting from $\iota(\gamma_{t_0})$ and ending in $U'$ is entirely contained in $U'$, cf.~Minguzzi \cite{minguzzi2019-causality}*{Thm.~4.44}. Since $\iota(\mms)$ is dense in $\mms'$, there exists a point $y\in\mms$ with $\iota(y) \in I'^+(\iota(\gamma_{t_0}))\cap U'$. Chronology-preservation of $\iota$ implies $\gamma_{t_0}\ll y$. This implies existence of an open convex neighborhood $\smash{O\subset I^-(y)}$ of $\gamma_{t_0}$. Continuity of $\gamma$ implies there is an open interval $I_1 \subset (0,1)$ around $t_0$ such that for every $t\in I_1 \cap [t_0,1)$, we have $\gamma_t\in O$. We now claim $\iota(\gamma_t)\in O'$ for every $t\in I_1 \cap [t_0,1)$. Indeed, we have $\gamma_{t_0} \ll \gamma_t \ll y$ by construction, hence $\smash{\iota(\gamma_{t_0}) \ll' \iota(\gamma_t) \ll'\iota(y)}$ by chronology-preservation of $\iota$.  Therefore, there exists a continuous timelike curve  starting in $\iota(\gamma_{t_0})$ through $\iota(\gamma_t)$ which ends in $\iota(y)$. Notably, since its endpoint lies in $U'$,  this curve does not leave $U'$. In particular, the inclusion $U'\subset O'$ entails $\iota(\gamma_t) \in O'$.

An analogous argument based on past  distinction of $\smash{(\mms',g')}$ yields the existence of an open interval $I_2\subset (0,1)$ around $t_0$ such that $\iota(\gamma_t)\in O'$ for every $t\in I_2 \cap (0,t_0]$.

Setting $I := I_1\cup I_2$ yields the desired interval above.
\end{proof}

The following  theorem is similar in spirit to a recent result of  Braun--Sämann \cite{braun-samann+}*{Lem. 2.13}. They give  such a result for  distance-preserving maps between Minguzzi--Suhr's bounded Lorentzian metric spaces \cite{minguzzi-suhr2022}, while ours is about order-preservation. 

\begin{theorem}[Existence of extension]\label{Th:Ext} Let $(\mms,g)$ and $(\mms',g')$ be two causally continuous  spacetimes, where $(\mms',g')$ is also future chronocomplete. Given any  dense set $D\subset\mms$,  also suppose $\iota\colon D\to \mms'$ is chronology-preserving with dense image. Then there exists a   chronology-preserving map that extends $\iota$ to all of $\mms$. 

If $(\mms,g)$ is also future  chronocomplete, this  extension is a chronological isomorphism.
\end{theorem}

Recall causal continuity means distinction plus reflectivity, cf.~\cref{Def:Causalcontinuity}.

\begin{proof}[Proof of \cref{Th:Ext}]  Let $x\in \mms\setminus D$. As $D$ is dense in $\mms$, we find $\smash{x^-,x^+ \in D}$ such that $\smash{x^- \ll x \ll x^+}$.  Owing to first countability and Hausdorffness of $\mms$, let $(N_i)_{i\in\N}$ constitute a nonincreasing  sequence of open neighborhoods of $x$ with  $\smash{\bigcap_{i\in\N}N_i = \{x\}}$. Define $\smash{x_0 := x^-}$. Clearly, the open set $I(x_0,x) \cap N_1$ is nonempty.  By density of $D$ in $\mms$, it contains a point $x_1\in D$ we fix. Inductively, we assume $x_i \in I(x_{i-1},x) \cap N_i \cap D$ is given for $i\in\N$. We analogously obtain the existence of a point $x_{i+1} \in  I(x_i,x) \cap N_{i+1} \cap D$ we fix. This way, we construct a chronologically increasing  sequence $(x_i)_{i\in\N}$ in $\mms$ converging to $x$.

By chronology-preservation of $\iota$, the sequence $(\iota(x_i))_{i\in\N}$ is  chronologically increasing  in $\mms'$. We also have $\smash{\iota(x_i) \ll'  \iota(x^+)}$ for every $i\in\N$. By future chronocompleteness, the sequence in question thus converges to a point $\smash{x' \in \mms'}$. We define $\smash{\iota(x) := x'}$. Clearly, this defines a (nonrelabeled) map $\iota\colon\mms\to \mms'$ which extends $\iota$ beyond $D$.

We claim $\iota$ is chronology-preserving. Let $x,y\in \mms$. If $x,y\in D$, we are done. Now suppose $x,y\in\mms\setminus D$. Assume $x\ll y$. Let $(x_i)_{i\in\N}$ and $(y_i)_{i\in\N}$ be sequences associated to $x$ and $y$ as  above, respectively. By density of $D$ in $\mms$, there exist $\smash{x^+,y^-\in D}$ such that $\smash{x \ll x^+ \ll y^-\ll y}$. Then for every sufficiently large $i\in\N$, we have $\smash{x_i\ll x^+ \ll y^- \ll y_i}$. By chronology-preservation, we have $\smash{\iota(x_i)\ll' \iota(x^+) \ll' \iota(y^-) \ll' \iota(y_i)}$. This shows that $\smash{\iota(x) \in \cl\,J'^-(\iota(x^+))}$ and thus $\smash{\iota(x^+)\in \cl \,J'^+(\iota(x))}$ since $(\mms,g)$ is future reflecting. This results in $\smash{\iota(x)\leq_\rmf' \iota(x^+)}$. Analogously, we prove $\smash{\iota(y^-) \leq_\rmp' \iota(y)}$ based on  past reflectivity of $(\mms',g')$. This yields the chain $\smash{\iota(x) \leq_\rmf'\iota(x^+) \ll \iota(y^-)\leq_\rmp'\iota(y)}$. Applying \cref{Le:Push-up} twice then entails $\iota(x)\ll\iota(y)$, as desired. Conversely, if $\iota(x) \ll'\iota(y)$, repeating this argument ---  using  density of $\iota(D)$ in $\mms'$ and the reflectivity  of $(\mms,g)$ --- establishes $x\ll y$. The remaining cases $x\in \mms\setminus D$ and $y\in D$ or $x\in D$ and $y\in\mms\setminus D$ are  shown  analogously; here, the argument is  simpler by taking the relevant sequence to be constant. 

We turn to the last statement. As $(\mms,g)$ is distinguishing as well, injectivity of the extension $\iota\colon\mms\to\mms'$ follows from the third paragraph and \cref{Le:Inj}. Surjectivity is argued as follows. The hypotheses and \cref{Pr:Uniq} imply $\iota$ necessarily  satisfies \eqref{Eq:Equ!}. Let $x'\in\mms'$; it suffices to address the case  $x'\in\mms'\setminus D$. By density of $\iota(D)$ in $\mms'$, we fix $\smash{x'^-,x'^+\in \iota(D)}$ with $\smash{x'^-\ll'  x'\ll' x'^+}$. The construction from the first paragraph yields a chronologically increasing sequence $\smash{(x'_i)_{i\in\N}}$ in $\iota(D)$ converging to $x'$. We shall write $\smash{\iota(x^+) := x'^+}$ and $\smash{\iota(x_i) := x_i'}$ for certain $\smash{x^+,x_i\in D}$, where $i\in\N$. Thanks to chronology-preservation, $(x_i)_{i\in\N}$ is a chronologically increasing sequence in $D$ with $\smash{x_i \ll x^+}$ for every $i\in\N$. The hypothesized future chronocompleteness of $(\mms,g)$ implies the existence of a limit $x\in\mms$ of $\smash{(x_i)_{i\in\N}}$.  Then $\iota(x) = x'$ by  \eqref{Eq:Equ!}, as desired.
\end{proof}

\begin{corollary}[Existence and uniqueness of extension]\label{Cor:ExUn} We suppose $(\mms,g)$ and $(\mms',g')$ are causally continuous and future chronocomplete spacetimes. Let $\iota\colon D\to\mms'$ form a chronology-preserving map with dense image, where $D\subset\mms$ is dense. Then $\iota$ constitutes the restriction of a unique chronological isomorphism; the latter satisfies \eqref{Eq:Equ!}.
\end{corollary}

Recall every globally hyperbolic spacetime is causally continuous and chronocomplete by \cref{Le:GHtoFC}. Thus,  \cref{Cor:ExUn} implies the following.

\begin{corollary}[Extension under global hyperbolicity] Let  $(\mms,g)$ and $(\mms',g')$ constitute two globally hyperbolic  spacetimes. Suppose moreover that $\iota\colon D\to \mms'$ is chronology-preserving with dense image, where  $D\subset\mms$ is dense. Then there is a unique chronological isomorphism extending $\iota$ beyond $D$; this extension obeys \eqref{Eq:Equ!}.
\end{corollary}

\subsection{Preparation for ``number''. Generic sequences} Recall the modeling from \cref{Sub:Contri}. In particular, recall $\meas$ denotes the normalized volume measure \eqref{Eq:NorVol} induced by $(\mms,g)$.

The material we collect now will  ensure later that the maps we  construct are in fact volume-preserving. The   considerations to follow apply to every separable metric space and do not involve spacetime geometry at all. 

Following Gromov \cite{gromov1981}*{§3$\smash{\frac{1}{2}}$.22}, we define the set 
\begin{align*}
G &:= \Big\{(x_i)_{i\in\N} \in \mms^\infty : \lim_{n\to\infty} \frac1n\big[\varphi(x_1) + \dots + \varphi(x_n)\big] = \int_\mms\varphi\d\meas\\
&\qquad\qquad \textnormal{for every bounded and continuous  function }\varphi\colon\mms\to\R\Big\}.
\end{align*}
The elements of $G$ is called \emph{generic sequences}.  Clearly, $(x_i)_{i\in\N}$ is a generic sequence in $\mms$ if and only if the sequence $(\meas_n)_{n\in\N}$ converges narrowly to $\meas$ in $\Prob(\mms)$, where 
\begin{align}\label{Eq:Narrowly}
\meas_n := \frac1n\big[\delta_{x_1} + \dots + \delta_{x_n}\big].
\end{align}

\begin{remark}[Modification of generic sequences] Genericity of a sequence is not altered by finite permutations of its elements or removal of finitely many points.
\end{remark}

\begin{lemma}[Density]\label{Le:Dens} Every generic  sequence $(x_i)_{i\in\N}$ in $\mms$ is dense in $\mms$. 
\end{lemma}

\begin{proof} We abbreviate $D := \{x_i : i\in\N\}$. Given any $x\in\mms$ and any $r>0$, it suffices to show $D\cap B_r(x) \neq\emptyset$, where $B_r(x)$ denotes the open Riemannian ball of radius $r$ centered at $x$. Consider the bounded and continuous function $\varphi\colon\mms\to \R$ with $\smash{\varphi:= (r-\met(x,\cdot))^+}$, where $\met$ is  the Riemannian distance induced by the tacit  Riemannian metric on $\mms$. As $\varphi$ is non\-negative on $\mms$ and bounded from below by $r/2$  on $\smash{B_{r/2}(x)}$  and  as $\smash{\vol_g[B_{r/2}(x)] >0}$ by the full support of $\smash{\vol_g}$, we easily obtain $\smash{\int_\mms \varphi\d\meas > 0}$. Using $(x_i)_{i\in\N}$ is generic, this implies infinitely many $i\in\N$ satisfy $\varphi(x_i) > 0$ and thus $x_i\in B_r(x)$.
\end{proof}

The next step consists in showing the existence of a generic sequence in $\mms$. For this, we need the following lemma, which is no more than general probability theory.

\begin{lemma}[Countable intersection]\label{Le:Count!} There exists a countable class $\scrC$ of bounded and continuous functions on $\mms$ such that $\smash{G = \bigcap_{\varphi\in\scrC} G_\varphi}$, where
\begin{align*}
G_\varphi := \Big\{(x_i)_{i\in\N}\in\mms^\infty : \lim_{n\to\infty} \frac1n \big[\varphi(x_1) + \dots + \varphi(x_n)\big] = \int_\mms\varphi\d\meas\Big\}.
\end{align*} 

In particular, the totality of generic sequences in $\mms$ is $\meas^{\otimes\infty}$-measurable.
\end{lemma}

\begin{proof} The last statement is a trivial consequence of the first claim.

To show the first claim, recall $G$ equals the set of all sequences $(x_i)_{i\in\N}$ in $\mms$ such that $(\meas_n)_{n\in\N}$ converges narrowly to $\meas$ in the notation of \eqref{Eq:Narrowly}. Since $\mms$ is separable, a general fact, cf.~e.g.~Ambrosio--Gigli--Savaré \cite{ambrosio-gigli-savare2008}*{p.~107}, asserts the narrow topology of $\Prob(\mms)$ is already determined by a countable family $\scrC$ of bounded and continuous functions. In particular, this implies $\smash{G= \bigcap_{\varphi\in\scrC} G_\varphi}$, which  is the desired identity.
\end{proof}

\begin{corollary}[Conegligibility]\label{Le:Coneg} We have   $\smash{\meas^{\otimes \infty}[G]=1}$.
\end{corollary}

\begin{proof} We are exceptionally deviating a little from our notation and define a probability space $(\Omega,\scrF,\PPP)$ as follows. Let $\Omega := \mms^\infty$, $\scrF$ be the completion of the $\sigma$-algebra on $\mms^\infty$ induced by  the Borel $\sigma$-algebra on $\mms$, and $\PPP := \meas^{\otimes\infty}$. Let $\varphi\colon \mms \to \R$ be a bounded and continuous function. Define a sequence  $\smash{(X_i^\varphi)_{i\in\N}}$ of random variables $X_i\colon\Omega\to\R$ by $\smash{X_i^\varphi := \varphi\circ\pr_i}$, where $\pr_i \colon\mms^\infty \to\mms$ denotes the projection onto the $i$-th coordinate of its argument. Then the random variables $\smash{(X_i^\varphi)_{i\in\N}}$ are clearly i.i.d.~and $\PPP$-integrable (as $\varphi$ is bounded). The strong law of large numbers implies $\smash{\PPP[G_\varphi]=1}$ with the $\smash{\scrB^{\otimes\infty}}$-measurable set $G_\varphi$ from \cref{Le:Count!}. The same lemma states $G$ is the countable intersection of a family of such  sets. This implies $\PPP[G]=1$, as desired.
\end{proof}

\subsection{``Order''. Conformal isometry by chronological isomorphy}\label{Sub:RDR} Now we turn to the ``order'' component of the proof of \cref{Th:Main}, ``if'' implication. It will imply  the two spacetimes in question are smoothly conformally isometric. 

We consider the projective family $\smash{(\nu^k)_{k\in\N}}$  from \eqref{Eq:nun}  with respect to $(f_k)_{k\in\N}$, where the map $\smash{f_{k+1}\colon \{0,1\}^{(k+1)\times (k+1)} \to \{0,1\}^{k\times k}}$ removes the last column and row of its input. Let $\smash{\{0,1\}^{\infty\times \infty} := \lim_{\leftarrow}\, \{0,1\}^{\bullet \times \bullet}}$ 
 be the standard Borel space of infinite $0$-$1$-matrices with the induced projection maps $\smash{\pr_k\colon \{0,1\}^{\infty\times\infty} \to \{0,1\}^{k\times k}}$, where $k\in\N$, and the induced projective $\sigma$-algebra, cf.~\cref{Sub:Kolmo}. By Kolmogorov's extension \cref{Th:Kolmogorov}, there is a unique Borel probability measure $\nu^\infty$ on $\smash{\{0,1\}^{\infty\times\infty}}$ such that $\smash{(\pr_k)_\push\nu^\infty = \nu^k}$ for every $k\in\N$. 

Uniqueness and the hypothesis $\smash{\nu^k= \nu'^k}$ for every $k\in\N$ from \cref{Th:Main} yield
\begin{align}\label{Eq:nununu}
\nu^\infty = \nu'^\infty.
\end{align}
On the other hand, define the $\smash{\meas^{\otimes \infty}}$-measurable map $\smash{\sfC^\infty\colon \mms^\infty \to \{0,1\}^{\infty\times\infty}}$ by
\begin{align*}
\sfC^\infty((x_i)_{i\in\N})_{kl} := \begin{cases} 1 & \textnormal{if }x_k \ll x_l,\\
0 & \textnormal{otherwise}.
\end{cases}
\end{align*}
Trivially, we then have $\smash{(\pr_k)_\push \sfC^\infty_\push\meas^{\otimes \infty} = \nu^k}$ for every $k\in\N$. By the uniqueness part of Kolmogorov's extension \cref{Th:Kolmogorov}, this implies $\smash{\nu^\infty = \sfC^\infty_\push\meas^{\otimes \infty}}$. Hence, by since the the sets $G$ and $G'$ of generic sequences have full measure with respect to $\smash{\meas^{\otimes\infty}}$ and $\smash{\meas'^{\otimes\infty}}$, respectively, by \cref{Le:Coneg}, from \eqref{Eq:nununu} we obtain
\begin{align}\label{Eq:Full}
 \sfC^\infty_\push\meas^{\otimes\infty} = \sfC'^{\otimes\infty}_\push\meas'^{\otimes\infty}.
\end{align}

The sets $\smash{\sfC^\infty(G)}$ and $\smash{\sfC'^\infty(G')}$ are measurable with respect to the measure \eqref{Eq:Full}, cf.~e.g. Srivastava \cite{srivastava1998}*{Thm.~4.3.1}. By \cref{Le:Coneg}, both  have full measure with respect to the measure \eqref{Eq:Full}; in particular, their intersection  is nonempty. Hence, there exist generic sequences $\smash{(x_i)_{i\in\N}}$ and $\smash{(x_i')_{i\in\N}}$ in $\mms$ and $\mms'$, respectively, with the property
\begin{align*}
\sfC^\infty((x_i)_{i\in\N}) = \sfC'^\infty((x_i')_{i\in\N}).
\end{align*}
In other words, every $\smash{i,j\in\N}$, we have $\smash{x_i\ll x_j}$ if and only if $\smash{x_i' \ll x_j'}$.

Set $\smash{D:= \{x_i :i\in\N\}}$ and $\smash{D' :=\{x_i' : i\in\N\}}$. We define chronology-preserving maps $\iota\colon D\to D'$ and $\kappa\colon D'\to D$ in the evident way by $\smash{\iota(x_i) := x_i'}$ and $\smash{\kappa(x_i') := x_i}$. We first claim these maps are well-defined. Indeed, suppose for some distinct $i,j\in\N$ we would  have $\smash{x_i = x_j}$ yet $\smash{x_i' \neq x_j'}$. As $(\mms',g')$ is distinguishing and $D'\setminus\{x_i', x_j'\}$ is dense in $\mms'$ by \cref{Le:Dens}, there exists $\smash{k\in\N}$ such that  $\smash{x_k'\in D'\setminus\{x_i', x_j'\}}$ is contained in one of the following sets: $\smash{I'^+(x_i')\setminus I'^+(x_j')}$, $\smash{I'^+(x_j')\setminus I'^+(x_i')}$, $\smash{I'^-(x_i')\setminus I'^-(x_j')}$, or $\smash{I'^-(x_j')\setminus I'^-(x_i')}$. We will assume the first inclusion; the proof is analogous in the other cases. This yields $\smash{x_i'\ll' x_k'}$ yet $\smash{x_j'\not\ll' x_k'}$ and therefore $\smash{x_i\ll x_k}$ yet $\smash{x_j\not\ll x_k}$ by the construction of the two generic sequences, which clearly contradicts the assumption $\smash{x_i=x_j}$. Analogously, we use  the distinction property of $(\mms,g)$ yields well-definedness of $\iota$.

With our hypotheses on $(\mms,g)$ and $(\mms',g')$,  \cref{Le:Dens}, and \cref{Cor:ExUn}, the above discussion yields $\iota$ and $\kappa$ uniquely extend to (nonrelabeled) chronological isomorphisms $\iota\colon \mms\to\mms'$ and $\kappa'\colon\mms'\to\mms$. In addition, the  identities $\smash{\iota\circ\kappa(x_i) = x_i}$ and $\smash{\kappa\circ\iota(x_i') = x_i'}$ for every $i\in\N$ combine with \eqref{Eq:Equ!} to yield $\smash{\kappa = \iota^{-1}}$.

By \cref{Th:Malament}, this implies $\iota$ (hence $\kappa$)  is a homeomorphism and a smooth conformal isometry. Let $\Sigma\colon\mms\to\R$ be a smooth nowhere vanishing function with $\smash{\iota^*g' = \Sigma^2\,g}$.

\subsection{``Number''. Isometry by volume-preservation}\label{Sub:NBR} Finally, we turn to the ``number'' ingredient of the proof of \cref{Th:Main}, ``if'' implication, which improves smooth conformal isometry to smooth isometry. To this aim, we claim $\kappa$ (hence $\iota$) is in fact volume-preserving. Combined with the above  considerations, this will force $\iota$ and $\kappa$ to be smooth isometries.

Let $\varphi \colon\mms\to\R$ be bounded and continuous. Since  the function $\varphi\circ\kappa$ is also bounded and continuous, applying genericity twice entails
\begin{align*}
\int_\mms\varphi\d\vol_g &= \lambda\int_\mms \varphi\d\meas \\
&=\lambda\lim_{n\to\infty} \frac{1}{n}\big[\varphi(x_1) + \dots + \varphi(x_n)\big]\\
&=\lambda\lim_{n\to\infty} \frac1n\big[\varphi\circ\kappa(x_1') +\dots+\varphi\circ\kappa(x_n')\big]\\
&=\lambda\int_\mms\varphi\circ\kappa\d\meas'\\
&=\int_{\mms}\varphi\d\kappa_\push\vol_{g'}.
\end{align*}
The arbitrariness of $\varphi$ yields $\smash{\kappa_\push\vol_{g'} =\vol_g}$, i.e.~$\kappa$ is volume-preserving.

On the other hand, the identities from \cref{Le:Trafo} imply
\begin{align*}
\vol_g = \kappa_\push\vol_{g'} = \vol_{\iota^*g'} = \vol_{\Sigma^2\,g} = \Sigma^d\,\vol_g,
\end{align*}
where we recall $d$ denotes the dimension of $\mms$ and $\mms'$. This yields $\Sigma = 1$ everywhere on $\mms$ since $\smash{\vol_g}$ has full support and $\Sigma$ is smooth. Thus, $\kappa$ (hence $\iota$) is a measure-preserving smooth isometry. This establishes \cref{Th:Main}.

\begin{remark}[Relation of \cref{Th:Main} to Bombelli's conjecture and the Hauptvermutung from CST]\label{Re:RelationBCHV} The identity   \eqref{Eq:Hyyy} is arguably quite strong: it stipulates the same chronological orders in law \emph{without} permutations of the random variables. Bombelli's conjecture \cite{bombelli2000} suggests that  the conclusion from \cref{Th:Main}, ``if'' implication,   is true \emph{with} permutations. More precisely, given any $k\in\N$ let $\sim_k$ designate the equivalence relation on $\smash{\{0,1\}^{k\times k}}$ by conjugacy with permutation matrices. Let $\smash{\pi_k\colon\{0,1\}^{k\times k} \to \{0,1\}^{k\times k}/\sim_k}$ denote the induced projection map. The conjecture states that if $\smash{(\pi_k)_\push\nu^k = (\pi_k)_\push\nu'^k}$ for every $k\in\N$, then $(\mms,g)$ and $(\mms',g')$ are smoothly isometric; here, $\smash{\nu^k}$ is from \eqref{Eq:nun}.

There are two challenges that need to be understood before  applying our argument to Bombelli's conjecture. The first is  $\smash{\{(\pi_k)_\push\nu^k: k\in\N\}}$ does not define a projective family of probability measures. It is thus unclear how to ``pass to the limit $k\to\infty$'' in the identity $\smash{(\pi_k)_\push\nu^k = (\pi_k)_\push\nu'^k}$ in order to construct two dense (hence infinite) and chronologically isomorphic sets. The second is  genericity of a sequence of random variables is not stable under random  permutation of its elements. 

An affirmative answer to Bombelli's conjecture would confirm \eqref{Eq:ONG} ``statistically''. A necessarily stronger ``point\-wise'' form of \eqref{Eq:ONG} would probably give a strong hint towards the major open problem of CST: the rigorous formulation of its Hauptvermutung, envisioned by Bombelli--Lee--Meyer--Sorkin \cite{bombelli-lee-meyer-sorkin1987}. We refer the reader to Surya \cite{surya2019}*{§3.1} for details. First attempts --- yet without a concrete proposal --- were made by Bombelli \cite{bombelli2000}, Noldus \cite{noldus2004}, Bombelli--Noldus \cite{bombelli-noldus2004}, and Bombelli--Noldus--Tafoya  \cite{bombelli-noldus-tafoya2012+}, and more recently Müller \cite{muller2025+} and Mondino--Sämann \cite{mondino-samann2025+}.
\end{remark}

\subsection{Modifications in the weighted case} In a similar manner, we can now address the proof of the ``if'' part of \cref{Th:Weighted}. In this situation, recall we replace the normalized volume measure \eqref{Eq:NorVol}  by the normalized weighted volume measure \eqref{Eq:NorVolII}. Following the lines of  \cref{Sub:RDR} verbatim, we obtain the existence of a diffeomorphism $\iota\colon\mms\to\mms'$ and a smooth nowhere vanishing function $\Sigma\colon\mms\to\R$ such that $\smash{\iota^*g' = \Sigma^2\,g}$. As in \cref{Sub:NBR}, we see $\smash{\iota^{-1}}$ (hence $\iota$) is measure-preserving.

Combining this with \cref{Le:Trafo} then yields 
\begin{align*}
\rme^V\,\vol_g = \iota^{-1}_\push\big[\rme^{V'}\,\vol_{g'}\big] = \rme^{V'\circ \iota}\,\vol_{\iota^*g'} = \rme^{V'\circ\iota}\,\vol_{\Sigma^2\,g} = \rme^{V'\circ\iota}\,\Sigma^d\,\vol_g.
\end{align*}
As in \cref{Sub:NBR}, this yields the desired relation \eqref{Eq:Rel} and establishes \cref{Th:Weighted}.

\bibliographystyle{amsrefs}

\end{document}